\documentclass{article}
\usepackage{fullpage}
\usepackage{booktabs} 
\usepackage[ruled]{algorithm2e} 
\usepackage{enumitem}
\usepackage{breakcites}
\usepackage[numbers]{natbib}
\usepackage{authblk}

\SetAlFnt{\small}
\SetAlCapFnt{\small}
\SetAlCapNameFnt{\small}
\SetAlCapHSkip{0pt}
\IncMargin{-\parindent}

\usepackage{enumerate}
\usepackage{hyperref}
\hypersetup{
    colorlinks = true,
 	allcolors = teal,
}

\usepackage{graphicx}
\usepackage{caption}
\usepackage{mathtools}
\usepackage{comment}
\usepackage{algorithmic}
\usepackage[mathscr]{eucal}
\usepackage{xcolor,colortbl}
\usepackage{amsmath, amssymb, amsthm}
\usepackage{color} 
\usepackage{cleveref}
\usepackage{nth}
\usepackage{multirow}
\usepackage{nicefrac}
\usepackage{subfig}
\definecolor{airforceblue}{rgb}{0.36, 0.54, 0.66}
\newcounter{note}[section]

\theoremstyle{definition}
\newtheorem{theorem}{Theorem}

\newtheorem{lemma}{Lemma}
\newtheorem{proposition}{Proposition}
\newtheorem{definition}{Definition}

\newtheorem{corollary}{Corollary}

\theoremstyle{remark}

\newtheorem{remark}{Remark}

\renewcommand{\le}{\leqslant}
\renewcommand{\leq}{\leqslant}
\renewcommand{\ge}{\geqslant}
\renewcommand{\geq}{\geqslant}

\newcommand{\alts}{\mathcal{A}}
\newcommand{\ags}{\mathcal N}
\newcommand{\X}{X}
\newcommand{\Y}{Y}
\newcommand{\W}{W}
\DeclarePairedDelimiter{\set}{\{}{\}}

\DeclareMathOperator{\E}{\mathbb{E}}
\newcommand{\bbR}{\mathbb{R}}
\newcommand{\calL}{\mathcal{L}}
\newcommand{\mcns}{\, \rhd_{\mathsf{M}}\, }
\newcommand{\ucns}{\,\rhd_{\mathsf{U}}\,}
\newcommand{\plu}{\text{plu}}

\newcommand{\mdist}{\mathsf{\dist^{ M}}}
\newcommand{\udist}{\mathsf{\dist^U}}

\newcommand{\PPV}{\ensuremath{\mathsf{PPV}}\xspace}

\newcommand{\pref}{\ensuremath{\sigma}}
\newcommand{\prefp}{\ensuremath{\vec{\pref}}}
\newcommand{\utp}{\ensuremath{\vec{u}}}

\newcommand{\scf}[2]{\ensuremath{\mathsf{SC}(#1, #2)}}

\newcommand{\swf}[2]{\ensuremath{\mathsf{SW}(#1, #2)}}
\newcommand{\swfs}[3]{\ensuremath{\mathsf{SW}_{#1}(#2, #3)}}

\newcommand{\wf}{w}

\newcommand{\oc}{\X^*}
\newcommand{\md}{\mathsf d^\mathsf{M}}
\newcommand{\ud}{\mathsf d^\mathsf{U}}
\newcommand{\p}{p}
\newcommand{\wfx}{\hat \wf}
\newcommand{\wfxp}{\hat \wf^+}

\newcommand{\dist}{\operatorname{\mathsf{dist}}}

\newcommand{\Su}{\mathsf S}
\newcommand{\su}{\mathsf s}
\newcommand{\ex}{\frac{\varepsilon}{6}}

\newcommand{\RandMech}{Truncated Harmonic\xspace}
\newcommand{\TRandMech}{Top-$t$ Truncated Harmonic\xspace}

\renewcommand{\epsilon}{\varepsilon}
\renewcommand{\hat}{\widehat}
\renewcommand{\bar}{\overline}
\renewcommand{\tilde}{\widetilde}

\DeclareMathOperator{\argmax}{\arg\max}
\DeclareMathOperator{\argmin}{\arg\min}

\usepackage{xspace}

\begin{document}\allowdisplaybreaks
\date{}
\title{Best of Both Distortion Worlds}
 \author[1]{Vasilis Gkatzelis}
	\author[2]{Mohamad Latifian}
 \author[2]{Nisarg Shah}
	\affil[1]{Drexel University}\affil[2]{University of Toronto} 

\maketitle
\begin{abstract}
We study the problem of designing voting rules that take as input the ordinal preferences of $n$ agents over a set of $m$ alternatives and output a single alternative, aiming to optimize the overall happiness of the agents. The input to the voting rule is each agent's ranking of the alternatives from most to least preferred, yet the agents have more refined (cardinal) preferences that capture the intensity with which they prefer one alternative over another. To quantify the extent to which voting rules can optimize over the cardinal preferences given access only to the ordinal ones, prior work has used the \emph{distortion} measure, i.e., the worst-case approximation ratio between a voting rule's performance and the best performance achievable given the cardinal preferences. 

The work on the distortion of voting rules has been largely divided into two ``worlds'': \emph{utilitarian distortion} and \emph{metric distortion}. In the former, the cardinal preferences of the agents correspond to general utilities and the goal is to maximize a normalized social welfare. In the latter, the agents' cardinal preferences correspond to costs given by distances in an underlying metric space and the goal is to minimize the (unnormalized) social cost. Several deterministic and randomized voting rules have been proposed and evaluated for each of these worlds separately, gradually improving the achievable distortion bounds, but none of the known voting rules perform well in both worlds simultaneously.

In this work, we prove that one can in fact achieve the ``best of both worlds'' by designing new voting rules, both deterministic and randomized, that simultaneously achieve near-optimal distortion guarantees in both distortion worlds. We also prove that this positive result does not generalize to the case where the voting rule is provided with the rankings of only the top-$t$ alternatives of each agent, for $t<m$, and study the extent to which such best-of-both-worlds guarantees can be achieved.
\end{abstract}
\section{Introduction}
A lot of recent work on computational social choice has focused on evaluating the \emph{distortion} of voting rules. Informally, this captures the extent to which they can maximize the social welfare (or minimize the social cost) when making a decision using only limited information regarding the preferences of the agents. In the most fundamental setting, the voting rule is given a set $\alts$ of  $m$ alternatives and a set $\ags$ of $n$ agents, each with their own preferences over the alternatives, and it needs to select an alternative. Even though the preferences of the agents can be complicated, the vast majority of the voting rules used in practice ask each agent to report just their \emph{ranking} over the available alternatives (i.e., their \emph{ordinal} preferences), or even just a prefix that ranks their top-$t$ alternatives for some $t<m$ (e.g., in many elections, the agents are asked to select their top choice, second choice, and third choice only). Although this provides very useful guidance for the voting rule, it does not capture the intensity with which an agent may prefer one alternative over another: two agents with the same ranking can be quite different in terms of how much they prefer their top choice over their second one. This limits the ability of the voting rule to optimize cardinal objectives such as the social welfare or the social cost. The distortion measures the ratio between the objective value achieved by the rule on an instance and the best possible objective value in that instance, in the worst case over all instances~\cite{PR06}. Without appropriate modeling assumptions, it is impossible to achieve any bounded distortion using only the ordinal preferences. The distortion literature can be largely divided into two worlds that make different modeling assumptions: \emph{utilitarian} and \emph{metric}.

\vspace{5pt}
$\bullet$ The \emph{utilitarian distortion} approach makes no assumptions regarding the relative intensity of the agents' preferences, captured by general von Neumann-Morgenstern utilities. The agents' happiness is measured by their expected utility, and the goal of the voting rule is to maximize the \emph{normalized} social welfare (i.e., to maximize the sum of the agents' expected utilities after normalizing each agent's utility values over different alternatives to add up to $1$). This normalization serves two primary purposes: i) it introduces a sense of fairness, by normalizing all the agents' utilities down to the same scale, while maintaining the relative intensity of their preferences, and ii) it enables the voting rule to achieve bounded distortion guarantees that would otherwise be impossible. We refer the reader to the work of \citet{aziz2020justifications} for several additional justifications for this objective function. 

\vspace{5pt}
$\bullet$ The \emph{metric distortion} approach focuses on instances where the intensity of the agents' preferences can be captured by a distance function in a metric space, i.e., they obey the triangle inequality. A common motivating example is the classic facility location problem~\cite{anshelevich2021ordinal,filos2021approximate}: the agents lie in some metric space and a facility needs to be opened at a location within that same metric space. Each agent prefers locations that are closer to them and they rank the potential locations accordingly. Another motivating application comes from ranking political candidates: each voter’s preferences over different candidates are determined by their distance on different issues (which, e.g., can be captured as different dimensions). The metric space imposes enough structure on the underlying cardinal preferences of the agents to enable much stronger distortion guarantees whenever it holds, even without any normalization. 
\vspace{5pt}

Recent work in the distortion literature has analyzed many classic voting rules and introduced several new ones to achieve (near-)optimal distortion bounds for both deterministic and randomized voting rules, focusing either on the utilitarian or on the metric setting. However, none of these voting rules perform well on both settings simultaneously. 
Therefore, if some designer chooses a voting rule optimized for metric distortion, they should be confident that the metric assumption holds in the setting at hand, otherwise they may get very high distortion. If, on the other hand, they choose a voting rule optimized for utilitarian distortion and the metric assumption happens to hold in their application domain, they may be getting optimal utilitarian distortion guarantees but missing out on much stronger metric distortion guarantees. In fact, it can be non-trivial to detect in advance whether the metric structure holds or not, even for the two main motivating applications discussed above. For example, in facility location the triangle inequality may not apply if the agents care about other features of the potential locations beyond the distance (e.g., proximity to some other facility that the agent often visits, or whether it is on their way to work), or if the agent's cost grows non-linearly in the distance to the location. Similarly, when voting over political candidates, analogous issues may arise if other factors can significantly affect the voters' preferences, e.g., who their parents or friends vote for, or whether they trust some candidate more than another.
 
In this paper, our goal is to investigate whether there exist (deterministic and randomized) voting rules that simultaneously achieve near-optimal distortion guarantees in both utilitarian and metric worlds, thus providing an obvious choice for someone looking to minimize distortion. Such rules would be perfect for deployment within an automated decision-making tool, providing uninitiated users with the opportunity to input their rankings over some alternatives and find a low-distortion alternative without having to seek expert input regarding which rule to use. 

\newcommand{\STAB}[1]{\begin{tabular}{@{}c@{}}#1\end{tabular}}

\begin{table}[ht]
 \setlength\extrarowheight{7pt}
\small
\centering
\begin{tabular}{|c|c||c|c|} 
\hline
&\textbf{Rules}      & \textbf{Metric Distortion} & \textbf{Utilitarian Distortion}  \\ [1ex]
\hline \hline
\multirow{4}{*}{\STAB{\rotatebox[origin=c]{90}{\textbf{Deterministic}}}}
&Plurality   & $\Theta(m)$                   & $\mathbf{\Theta(m^2)}$                    \\
                         
&Copeland                & $\mathbf{\Theta(1)}$                         & Unbounded                        \\
&Plurality Matching, Plurality Veto          & $\mathbf{\Theta(1)}$      & $\Omega(n/m)$         \\[0.2cm]\cline{2-4}
&Pruned Plurality Veto\textsuperscript{*} & $\mathbf{\Theta(1)}$ & $\mathbf{\Theta(m^2)}$  \\[1ex] \hline \hline
\multirow{4}{*}{\STAB{\rotatebox[origin=c]{90}{\textbf{Randomized}}}}
&Random Dictatorship            & $\mathbf{\Theta(1)}$  & $\Theta(m\sqrt m)$     \\
&Harmonic Rule                  & Unbounded             & $\mathbf{\tilde \Theta(\sqrt{m})}$ \\
&Stable Lottery Rule            & Unbounded             & $\mathbf{\Theta(\sqrt{m})}$ \\[0.2cm]\cline{2-4}
&\RandMech Rule\textsuperscript{*}                 & $\mathbf{\Theta(1)}$  & $\mathbf{\tilde \Theta(\sqrt{m})}$\\[1ex]
\hline
\end{tabular}
\caption{Our novel voting rules, (deterministic) Pruned Plurality Veto and (randomized) Truncated Harmonic Rule, achieve asymptotically optimal metric and utilitarian distortions simultaneously (up to a log factor in one case). Existing rules that are asymptotically optimal in one of the two worlds are provided for comparison.}
\label{tbl:well-known}
\end{table}

\subsection{Our Results}
We propose deterministic and randomized voting rules that take agents' full rankings or top-$t$ preferences as input, return a single alternative, and simultaneously provide near-optimal metric and utilitarian distortion guarantees. Our proposed rules add a safety net by ensuring that even when the metric assumption does not hold, approximately optimal welfare will be achieved with respect to all possible utility functions consistent with the voters' ordinal preferences.

\paragraph{Full rankings.} For deterministic rules, the optimal guarantees we seek are $\Theta(1)$ metric distortion~\cite{GHS20} and $\Theta(m^2)$ utilitarian distortion~\cite{CP11,CNPS17}. For randomized rules, the optimal guarantees we seek are $\Theta(1)$ metric distortion~\cite{AP17} and $\tilde{\Theta}(\sqrt{m})$ utilitarian distortion~\cite{BCHL+15,EKPS22}. 

We first observe that all the existing voting rules that achieve the desired guarantee for either the utilitarian or the metric world fail to achieve the desired guarantee for the other world. 

Our main contribution is to deliver positive news: it is indeed possible to achieve the ``best of both worlds'' using both deterministic and randomized rules! Specifically, we design two novel voting rules: the deterministic \emph{Pruned Plurality Veto} rule achieves the desired $\Theta(1)$ metric distortion and $\Theta(m^2)$ utilitarian distortion simultaneously, and the randomized \emph{Truncated Harmonic} rule achieves $\Theta(1)$ metric distortion and $\tilde{\Theta}(\sqrt{m})$ utilitarian distortion simultaneously. Curiously, while Pruned Plurality Veto builds on Plurality Veto, which is optimal for metric distortion, Truncated Harmonic builds on the Harmonic Rule, which is optimal for utilitarian distortion. While both rules employ intuitive modifications, the proofs of their distortion guarantees require several new technical lemmas, which may be of independent interest.

\Cref{tbl:well-known} shows the comparison between our voting rules that achieve asymptotically optimal distortion under both metric and utilitarian worlds simultaneously, and the known voting rules that achieve asymptotically optimal distortion under only one of the two worlds.  

\paragraph{Top-$t$ preferences.} In stark contrast to the good news with full rankings, we deliver bad news when the input is top-$t$ preferences with small $t$.

For deterministic rules, the optimal utilitarian distortion for any $t$ is $\Theta(m^2)$~\cite{CP11,CNPS17}. We prove that any deterministic rule with bounded utilitarian distortion must incur $\Theta(m-t+1)$ metric distortion, as opposed to the optimal $\Theta(m/t)$ metric distortion~\cite{kempe2020communication}, which is sub-optimal when $t$ is neither $\Theta(1)$ nor $m-\Theta(1)$. 

For randomized rules, the optimal metric distortion is $\Theta(1)$ for any $t\geq 1$. We prove that $\Theta(1)$ metric distortion implies $\Omega(\max((m/t)^{1.5},\sqrt{m}))$ utilitarian distortion, which is strictly worse than the optimal $\Theta(\max(m/t,\sqrt{m}))$ utilitarian distortion~\cite{BHLS22} for all $t \in o(m^{2/3})$. 

Our results chart a more general trade-off between metric and utilitarian distortions, and we complement our negative results with novel constructive upper bounds.

\subsection{Related Work}
There is a huge body of work on distortion, from both metric and utilitarian points of view; see the recent survey of \citet{AFSV21} for a detailed summary. \citet{PR06} introduced the notion of distortion in the utilitarian setting.
\citet{CP11} proved that the Plurality rule has $\Theta(m^2)$ utilitarian distortion, which later proved to be optimal over all deterministic rules~\cite{CNPS17}. \citet{BCHL+15} studied the utilitarian distortion of randomized voting rules and proposed the Harmonic Rule, which we use in our work. They proved that its distortion is $O(\sqrt {mH_m})$ and designed a more complex rule with distortion $O(\sqrt {m} \log^*(m))$, also proving a lower bound of $\Omega(\sqrt{m})$. This gap was closed by \citet{EKPS22}, who achieved the optimal $\Theta(\sqrt{m})$ distortion by introducing the Stable Lottery Rule.

On the other end, \citet{ABEP+18} initiated the study of distortion in the metric setting. They proved that the Copeland rule 
has a metric distortion of $5$. They proved a lower bound of $3$ on the metric distortion of any deterministic voting rule and conjectured that this lower bound is tight. \citet{MW19} improved the upper bound to $4.236$, and \citet{GHS20} finally resolved the conjecture by achieving a metric distortion of $3$ by introducing the Plurality Matching rule. 
While the Plurality Matching rule is rather complicated, \citet{FD22} subsequently introduced a simpler refinement, Plurality Veto, which still achieves a metric distortion of $3$. Very recent work has provided an even deeper understanding of the class of deterministic rules achieving a metric distortion of 3~\cite{PracticalOptimalDistortion,peters2023note}.

Even though a metric distortion better than $3$ is not achievable by deterministic rules, \citet{AP17} proved that randomization can break this barrier. They proved that Random Dictatorship has metric distortion slightly better than $3$ and provide a lower bound of $2$ on the metric distortion of any randomized rule. Several works have since tried to close this gap, with \citet{charikar2022metric} recently improving the lower bound to $2.0261$. Nonetheless, the optimal metric distortion of randomized rules is still $\Theta(1)$. 

While the aforementioned work considers full rankings as input, there is significant interest in analyzing the distortion achievable using other types of information. \citet{ABFV21} considered eliciting restricted cardinal information regarding the agents' preferences in addition to their ordinal preferences. 
In contrast, some work considers eliciting \emph{less} information than ordinal preferences, focusing on top-$t$ preferences. In the metric setting, \citet{kempe2020communication} proved a lower bound of $\frac{2m-t}{t}$ and an upper bound of $\frac{79m}{t}$ on the best metric distortion that any deterministic voting rule can achieve using top-$t$ preferences; the upper bound was later improved to $6m/t+1$ by \citet{AFP22}. While the exact bound is yet to be identified, this pins down the asymptotic bound as $\Theta(m/t)$. For randomized rules, the best bound is still $\Theta(1)$ for any $t$, achieved by Random Dictatorship using only top-$1$ preferences (i.e., plurality votes). 
In the utilitarian setting, \citet{BHLS22} proved that the optimal utilitarian distortion achievable by randomized rules using top-$t$ preferences is $\Theta(\max(m/t,\sqrt{m}))$. For deterministic rules, it is $\Theta(m^2)$ for any $t$ because Plurality achieves it using top-$1$ preferences and is optimal even within rules using full rankings.

The distortion framework has also been extended to study many other problems such as the elicitation-distortion tradeoff~\cite{mandal2019thrifty,mandal2020optimal,kempe2020communication}, more complex elections~\cite{benade2021preference,borodin2019primaries,caragiannis2022metric,ebadian2023approval}, and even problems beyond voting~\cite{anshelevich2016blind,amanatidis2022few,halpern2021fair,ebadian2022efficient}.

\section{Preliminaries}
Let $\ags$ be a set of $n$ agents and $\alts$ be a set of $m$ alternatives. 
Throughout the paper, agents are labeled $i,j$ and alternatives are labeled $\X, \Y, \W$. We use $\calL(\alts)$ to denote the set of all rankings (linear orders) over the alternatives, and $\Delta(\alts)$ to denote the set of all distributions over $\alts$.

\paragraph{Preference profile.} Each agent $i \in \ags$ has a preference ranking $\pref_i \in \calL(\alts)$ over the alternatives. Collectively, they form a preference profile $\prefp = (\pref_1, \pref_2, \ldots, \pref_n)$. We use $r_i(\X)$ to denote the rank of alternative $\X$ in agent $i$'s ranking, and say agent $i$ prefers alternative $\X$ to alternative $\Y$ if $r_i(\X)<r_i(\Y)$, which we also denote as $\X \succ_i \Y$.
The plurality score $\plu(\X,\prefp)$ of alternative $\X$ under $\prefp$ is the number of agents who rank $\X$ first. We use $\ags^\X$ to denote the set of such agents. 

\paragraph{Voting Rule.} A (randomized) voting rule $f: \calL(\alts)^n \to \Delta(\alts)$ is a function that takes as input a preference profile $\prefp$ and outputs a probability distribution $f(\prefp)$ over alternatives. We call $f$ deterministic if it always outputs a single alternative with probability $1$; in this case, we slightly abuse the notation to let $f(\prefp)$ denote this alternative. The Plurality rule selects an alternative with the highest plurality score. 

We assume that the preference profile stems from underlying cardinal preferences of the agents over the alternatives. We consider two frameworks for modeling these cardinal preferences. In this paper, our goal is to design voting rules that perform well under both frameworks simultaneously and prove corresponding lower bounds on such ``best of both worlds'' guarantees.

\subsection{Metric Framework}
In the metric framework, we assume that all the agents and alternatives are embedded in a (pseudo) metric space identified by a distance function $d: (\ags \cup \alts)^2 \to \bbR_{\ge 0}$ satisfying the following.
\begin{itemize}
    \item $d(a,a) = 0$ for all $a \in \ags \cup \alts$.
    \item Symmetry: $d(a,b) = d(b,a)$, for all $a,b \in \ags \cup \alts$.
    \item Triangle inequality: $d(a,b) \leq d(a,c)+ d(c,b)$ for all $a,b,c \in \ags \cup \alts$.
\end{itemize}

We say that preference profile $\prefp$ is consistent with metric space $d$, denoted as $d \mcns \prefp$, if, for every $i\in \ags$ and $\X,\Y\in \alts$, we have $\X \succ_i \Y \Rightarrow d(i, \X) \leq d(i, \Y)$. In this framework, we capture the quality of each alternative by her total cost to the agents: define $\scf{\X}{d} = \sum_{i \in \ags} d(i, \X)$ be the social cost of alternative $\X$ in metric space $d$.

\paragraph{Metric Distortion.}
In this setting, the rule that selects an alternative with a lower social cost is generally considered more efficient. Formally, we define the metric distortion of a distribution over alternatives $p \in \Delta(\alts)$ in metric space $d$ as the ratio between the expected social cost of an alternative sampled from $p$ to the minimum possible social cost, i.e., 
$$\mdist(p, d) = \frac{\E_{\X \sim p} [\scf{\X}{d}]}{\min_{\X \in \alts} \scf{\X}{d}}.$$
The metric distortion of a voting rule $f$ on preference profile $\prefp$ is the worst-case metric distortion of $f(\prefp)$ on any metric $d \mcns \prefp$: 
$
\mdist(f,\prefp) = \sup_{d\, :\, d \mcns \prefp} \mdist(f(\prefp), d).
$
Finally, the metric distortion of $f$ is obtained by taking the worst case over all preference profiles: $\mdist(f) = \max_{\prefp \in \calL(\alts)^n} \mdist(f,\prefp)$.

\subsection{Utilitarian Framework}
The utilitarian framework assumes that each agent $i \in \ags$ has a utility function $u_i: \alts \to \bbR_{\ge 0}$ over the alternatives. Collectively, they form the utility profile $\utp = (u_1, u_2, \ldots, u_n)$. Following the literature~\cite{aziz2020justifications}, we assume that the utility functions are unit-sum, which means that for each agent $i$, we have $\sum_{\X \in \alts} u_i(\X) = 1$. 
We say that preference profile $\prefp$ is consistent with utility profile $\utp$, denoted $\utp \ucns \prefp$, if for each $i \in \ags$ and $\X,\Y \in \alts$, we have $\X \succ_i \Y \Rightarrow u_i(\X) \geq u_i(\Y)$. 

Here, the quality of each alternative is determined based on the sum of the utilities of the agents for it: define the social welfare of alternative $\X$ in utility profile $\utp$ as $\swf{\X}{\utp} = \sum_{i \in \ags} u_i(\X)$. We remark that studying the social welfare with unit-sum utilities is equivalent to studying the \emph{normalized} social welfare with unrestricted utilities, which is the formulation described in the introduction. We use the former to be consistent with the literature. Similarly to the metric case, we can define distortion in the utilitarian framework.

\paragraph{Utilitarian Distortion.}
The utilitarian distortion of distribution over alternatives $p \in \Delta(\alts)$ in utility profile $\utp$  is defined as:
$$\udist(p, \utp) = \frac{\max_{\X \in \alts} \swf{\X}{\utp}}{\E_{\X \sim p}[\swf{\X}{\utp}]}.$$
The utilitarian distortion of voting rule $f$ on preference profile $\prefp$ is the worst-case distortion of $f(\prefp)$ on any utility profile $\utp \ucns \prefp$:
$
\udist(f,\prefp) = \sup_{\utp\, :\, \utp \ucns \prefp} \udist(f(\prefp), \utp).
$ 
Finally, the utilitarian distortion of $f$ is obtained by taking the worst case over all preference profiles: $\udist(f) = \max_{\prefp \in \calL(\alts)^n} \udist(f,\prefp)$.

\section{A Deterministic Voting Rule}

Our goal in this section is to check if there exists a deterministic voting rule that achieves near-optimal distortion guarantees in both metric and utilitarian frameworks simultaneously. Let us first review the state-of-the-art guarantees in each framework. For deterministic rules, it is known that the best possible metric distortion is $3$~\cite{ABEP+18,GHS20}, which can be achieved using Plurality Matching~\cite{GHS20} or its refinement Plurality Veto~\cite{FD22}, and the best possible utilitarian distortion is $O(m^2)$, which can be achieved by Plurality~\cite{CP11,CNPS17}.

While the aforementioned rules achieve optimal distortion in one framework, we remark that they incur a terrible distortion in the other framework, thus failing to provide our desired best-of-both-worlds guarantee. Plurality Matching (and its refinement Plurality Veto) incurs a utilitarian distortion of $\Omega(n/m)$ (as opposed to the desired $O(m^2)$) and Plurality incurs a metric distortion of $\Omega(m)$ (as opposed to the desired $O(1)$). 

Because Plurality Matching and Plurality Veto are not central to our work, we refer an interested reader to the works of \citet{GHS20} and \citet{FD22} for their definitions.

\begin{proposition}
    For $m \ge 3$, the utilitarian distortion of Plurality Matching (and consequently, its refinement Plurality Veto) is $\Omega(n/m)$.
\end{proposition}
\begin{proof}
    Let $\alts = \set{X^*,X_1,X_2,\ldots,X_{m-1}}$. Consider a preference profile $\prefp$ in which, for $k \in [m-1]$, $n/(m-1)-1$ agents rank $X_k$ first, $X^*$ second, and $X_{k+1}$ last (with $X_m$ redefined as $X_1$ for cyclicity) and one agent ranks $X^*$ first and $X_k$ last. The rest of the preference profile can be completed arbitrarily. 
    
    \citet{GHS20} prove that an alternative can only be returned by Plurality Matching if it is ranked first at least as many times as it is ranked last. Note that every alternative except $X^*$ appears last once more than it appears first. Hence, Plurality Matching and its refinement Plurality Veto must output $X^*$. 

    Consider a utility profile $\utp \ucns \prefp$ in which every agent who ranks $X^*$ first has utility $1/m$ for all alternatives, and every other agent has utility $1$ for her top choice. Then, $\swf{X^*}{\utp} = O(1)$ whereas $\swf{X_1}{\utp} = \Omega(n/m)$, yielding the desired utilitarian distortion lower bound.
\end{proof}

\begin{proposition}[\cite{ABEP+18}]
    The metric distortion of the Plurality rule is $2m-1$.
\end{proposition}

This leaves open the question of whether there exists a deterministic rule that simultaneously achieves $O(1)$ metric distortion and $O(m^2)$ utilitarian distortion. 

\subsection{Pruned Plurality Veto}
We settle this question positively by designing a new deterministic voting rule, \emph{Pruned Plurality Veto}, and proving that it achieves the desired best-of-both-worlds guarantee. We remark that our rule does not critically rely on using Plurality Veto; any deterministic voting rule with $O(1)$ metric distortion can be used instead; the metric distortion of our rule will be a constant times greater than the metric distortion of this base rule. We use Plurality Veto because it is a simple voting rule that achieves the optimal metric distortion guarantee of $3$. 

\begin{definition}[Pruned Plurality Veto (\PPV)]
Given any $\varepsilon > 0$ and a preference profile $\prefp$, the (deterministic) Pruned Plurality Veto Rule $f^{\mathsf{\PPV}}$ computes the set of alternatives $\alts'$ with plurality score at least $\frac{\varepsilon n}{(6+\varepsilon)m}$ (note that $\alts'$ must be non-empty by the Pigeonhole principle), restricts $\prefp$ to the alternatives in $\alts'$ to obtain preference profile $\prefp'$, and runs Plurality Veto on $\prefp'$.
\end{definition}

The intuition behind the rule is simple. The argument used by \citet{CNPS17} for establishing $O(m^2)$ utilitarian distortion of Plurality continues to hold when selecting any alternative with plurality score $\Omega(n/m)$. Hence, limiting our attention to such alternatives easily takes care of the utilitarian distortion guarantee. The trick is to prove that simply running a voting rule with constant metric distortion on the preference profile restricted to such alternatives still yields (slightly higher) constant metric distortion. For this, we need two lemmas. 

The first one is due to \citet{ABEP+18}. We provide the short proof for completeness.
\begin{lemma}
\label{lem:YW-half}
    For any metric space $d$, alternatives $\X$ and $\Y$, and agent $i$ who prefers $\Y$ to $\X$ (i.e., $\Y \succ_i \X$), we have $d(i, \X) \geq d(\Y,\X)/2$.
\end{lemma}
\begin{proof}
    Note that
    \[
    d(\Y,\X) \le d(i,\Y) + d(i,\X) \le 2 d(i,\X),
    \]
    where the last inequality holds due to $\Y \succ_i \X$.
\end{proof}

The next lemma, which is the crux of our proof, shows that ignoring alternatives with low plurality score and running a voting rule on the remaining alternatives can only increase its metric distortion by a small factor. This lemma may be of independent interest.

\begin{lemma}\label{lem:subset-mdist}
    Consider any preference profile $\prefp$ and metric $d \mcns \prefp$. For some threshold $\tau < \nicefrac{1}{m}$, let $\alts' \subseteq \alts$ be the set of alternatives with plurality score at least $\tau n$. Then,
    \[
    \frac{\min_{\X \in \alts'} \scf{\X}{d}}{\min_{\X \in \alts} \scf{\X}{d}} \le 1+\frac{2}{1-\tau m}.
    \]
    Consequently, for a constant $\md$, applying a voting rule with a metric distortion of at most $\md$ on the preference profile obtained by restricting $\prefp$ to $\alts'$ incurs a metric distortion of at most $\md \cdot (1+\nicefrac{2}{(1-\tau m)})$. 
\end{lemma}
\begin{proof}
Let $\oc \in \argmin_{\X \in \alts} \scf{\X}{d}$ be an optimal alternative and $\X^+ \in \argmin_{\X \in \alts'} d(\X,\oc)$ be the closest alternative to $\oc$ in $\alts'$. Recall that for an alternative $\Y$, $\alts^{\Y}$ denotes the set of agents who rank $\Y$ first. For each alternative $\Y \notin \alts'$, $\plu(\Y,\prefp) \le \tau n$, so $\sum_{\Y \notin \alts'} |\alts^\Y| \le \tau m n$, implying  $\sum_{\Y \in \alts'} |\alts^\Y| \ge n \cdot (1-\tau m)$. 

We now prove a lower bound on the optimal social cost.
\begin{align*}
    \scf{\oc}{d} &= \sum_{i \in \ags} d(i, \oc)  \geq \sum_{\Y \in \alts'} \sum_{i \in \alts^{\Y}} d(i, \oc) \\
    &\geq \sum_{\Y \in \alts'} \sum_{i \in \alts^{\Y}} d(\oc, \Y)/2 & (\text{By \Cref{lem:YW-half}}) \\
    &\geq \sum_{\Y \in \alts'} |\alts^{\Y}| \cdot d(\oc, \X^+)/2 & (\because \text{$\X^+$ is the closest member of $\alts'$ to $\oc$})\\
    & \geq n \cdot (1-\tau m) \cdot d(\oc, \X^+)/2.
\end{align*}
On the other hand, using the triangle inequality for each agent and adding across agents, we get
$$\scf{\X^+}{d} \leq \scf{\oc}{d} + n \cdot d(\oc, \X^+).$$

Let $\X' \in \argmin_{\X \in \alts'} \scf{\X}{d}$. Then,
\begin{align*}
\frac{\scf{\X'}{d}}{\scf{\oc}{d}} & \leq \frac{\scf{\X^+}{d}}{\scf{\oc}{d}} &\left(\because \X' \in \argmin_{\X \in \alts'} \scf{\X}{d}\right)  \\
    & \leq \frac{ \scf{\oc}{d} + n \cdot d(\oc, \X^+)}{\scf{\oc}{d}} \\
    & \leq 1 + \frac{ n \cdot d(\oc, \X^+) }{n \cdot (1-\tau m) \cdot d(\oc, \X^+)/2} = 1 + \frac{2}{1-\tau m}.
\end{align*}

Next, let $\md$ be a constant and $f$ be a voting rule with a metric distortion of at most $\md$. Let $\prefp'$ be the preference profile restricted to alternatives in $\alts'$ and $\X = f(\prefp')$. Then, we have
\[
\scf{\X}{d} \le \md \cdot \scf{\X'}{d} \le \md \cdot \left(1+\frac{2}{1-\tau m}\right) \cdot \scf{\oc}{d},
\]
where the first transition is due to the metric distortion guarantee of $f$ and the second is proven above. 
\end{proof}

We are now ready to prove the distortion guarantees of Pruned Plurality Veto. 
\begin{theorem}\label{thm:bobw-det}
    Given any $\epsilon > 0$, Pruned Plurality Veto (\PPV) with parameter $\epsilon$ has a metric distortion of $9+\varepsilon$ and a utilitarian distortion of $O(m^2/\epsilon)$.
    
\end{theorem}

\begin{proof}
Let $\prefp$ be any preference profile. Let $\alts'$ be the set of alternatives with plurality score at least $\frac{\varepsilon n}{(6+\varepsilon)m}$, and $\bar{\alts} = \alts \setminus \alts'$. Note that \PPV always returns an alternative from $\alts'$.

This is sufficient for the utilitarian distortion guarantee, since any alternative in $\alts'$ is the top choice of at least $\frac{\varepsilon n}{(6+\varepsilon)m}$ agents, and each agent has utility at least $1/m$ for her top alternative due to the Pigeonhole principle. Hence, its social welfare is at least $\frac{\varepsilon n}{(6+\varepsilon) m^2}$. Due to unit-sum utilities, the maximum social welfare of any alternative is at most $n$. Hence, the utilitarian distortion of \PPV is at most $\frac{(6+\varepsilon)m^2}{\varepsilon}$.

For the metric distortion, we simply need to apply \Cref{lem:subset-mdist} with $\md = 3$ (the metric distortion of Plurality Veto) and $\tau = \frac{\varepsilon}{(6+\varepsilon) m}$. We get that the metric distortion of \PPV is at most
\[
\md \cdot \left(1+\frac{2}{1-\tau m}\right) = 3 \cdot \left(1+\frac{2}{1-\frac{\varepsilon}{6+\varepsilon}}\right) = 3 \cdot \left(3+\frac{\varepsilon}{3}\right) = 9 + \varepsilon.\qedhere
\]
\end{proof}

While Pruned Plurality Veto achieves asymptotically optimal distortion in both metric and utilitarian frameworks simultaneously, it does not achieve metric distortion better than $9$, when the optimal metric distortion is $3$. In the metric framework, since the optimal distortion is already a constant, it is common to care about what the exact constant is. Indeed, a metric distortion of $5$ was established through Copeland's rule already in the seminal paper that introduced this framework~\cite{ABEP+18}. Reducing this to $3$ was an open question that was settled recently~\cite{GHS20}. 

This might lead one to wonder whether one can obtain a tighter best-of-both-worlds guarantee with a utilitarian distortion of $O(m^2)$ and a metric distortion of at most $3$. We show that this is not possible and leave open the question of improving upon the metric distortion of $9+\epsilon$ in \Cref{thm:bobw-det}.

\begin{theorem}\label{thm:bobw-det-imposs}
    Any deterministic rule with a utilitarian distortion of $o(m^2 \sqrt{m})$ has a metric distortion strictly greater than $3$. 
\end{theorem}
\begin{proof}
    Let $\X$ be a specific alternative and partition the remaining alternatives into 3 sets, $\alts_1$, $\alts_2$, and $\alts_3$, each with $(m-1)/3$ alternatives. Furthermore, partition the agents into 3 sets, $\ags_1$ and $\ags_2$ with $n(1-1/\sqrt m)/2$ agents each, and $\ags_3$ with $n/\sqrt m$ agents.
    
    Consider the following preference profile $\prefp$. For $k \in [3]$, each alternative in $\alts_k$ appears as the top choice of $3/(m-1)$ fraction of the agents in $\ags_i$, and $\X$ appears as the second choice of every agent. Each agent prefers every alternative in $\alts_3$ to any alternative in $\alts_1 \cup \alts_2$ other than her top choice. Agents in $\ags_1$ prefer every alternative in $\alts_1$ to any alternative in $\alts_2$, and agents in $\ags_2$ prefer every alternative in $\alts_2$ to any alternative in $\alts_1$.

    Suppose $f$ is a deterministic rule with a utilitarian distortion of $o(m^2 \sqrt{m})$. Let $f(\prefp) = \X^*$. First, we prove that $\X^* \in \alts_1 \cup \alts_2$ on $\prefp$. 
    
    If $\X^* = \X$, consider the consistent utility profile in which each agent has utility $1$ for her top choice and zero for the rest. In this case, the utilitarian distortion of $\X$ is unbounded, which is a contradiction. If $\X^* \in \alts_3$, consider the consistent utility profile $\utp$ in which agents who have $\X^*$ as their top choice have utility $1/m$ for every alternative and the other agents have utility $1/2$ for their top choice and $1/2$ for $\X$. Then, $\swf{\X^*}{\utp} = O(n/(m^2 \sqrt{m}))$ whereas $\swf{\X}{\utp} = \Omega(1)$, yielding a utilitarian distortion of $\Omega(m^2 \sqrt{m})$, which is again a contradiction. Hence, we have $\X^* \in \alts_1 \cup \alts_2$.

    Without loss of generality, assume $\X^* \in \alts_1$. Consider a consistent one-dimensional distance metric $d$ over $\bbR$ in which alternatives in $\alts_1$ lie at $0$, agents in $\ags_1$ lie at $1$, and the remaining agents and alternatives lie at $2$. Thus, the distances are as follows: for all agents $i$ and alternatives $\Y$, 
    $$d(i, \Y) = \begin{cases}
    1 & i \in \ags_1 \\
    2 & i \notin \ags_1, \Y \in \alts_1 \\
    0 & i \notin \ags_1, \Y \notin \alts_1.
    \end{cases}.$$

    Now, for $\X^* \in \alts_1$, we have:
    $$\scf{\X^*}{d} = \frac{n}{2}\left(1-\frac{1}{\sqrt m}\right) + 2 \cdot \left(n - \frac{n}{2} \cdot \left(1-\frac{1}{\sqrt m}\right)\right) = n \cdot \frac{3\sqrt m + 1}{2\sqrt m}.$$
    On the other hand, 
    $$\scf{\X}{d} = 1 \cdot \frac{n}{2} \cdot \left(1-\frac{1}{\sqrt m}\right) = n \frac{\sqrt m - 1}{2\sqrt m}.$$
    Thus, we have 
    $$\mdist(f,\prefp) \geq \frac{\scf{\X^*}{d}}{\scf{\X}{d}} = \frac{3\sqrt m + 1}{\sqrt m - 1} > 3.\qedhere
    $$
\end{proof}

\section{A Randomized Voting Rule}

Next, we turn to randomized voting rules. In the metric world, the best distortion achievable by a randomized rule is known to be in $[2.0261,3)$~\cite{charikar2022metric,AP17}, and closing that gap remains a major open question.

There exist randomized voting rules that guarantee a distortion slightly better than $3$ (but their distortion converges to 3 as the number of agents or alternatives increases)~\cite{AP17,GHS20,kempe2020communication}, and notable among them is the Random Dictatorship, which returns the top alternative of an agent selected uniformly at random.

In the utilitarian world, the best achievable distortion by a randomized rule was recently shown to be $\Theta(\sqrt{m})$~\cite{EKPS22} using the Stable Lottery Rule, improving upon a  bound of $\Theta(\sqrt{m \log m})$ via the Harmonic Rule. 

Like in the case of deterministic rules, one might wonder whether these rules, which achieve near-optimal distortion in one framework, also achieve near-optimal distortion in the other framework, thus achieving the best of both worlds. Unfortunately, this is yet again not true.

We omit defining the Stable Lottery Rule and the Harmonic Rule, but it suffices to know that they assign a positive probability to every alternative. In a preference profile where an alternative $\X$ is ranked last by every agent, it could have arbitrarily large social cost, and hence, assigning any positive probability to it yields unbounded metric distortion.

\begin{proposition}
    The Stable Lottery Rule and the Harmonic Rule have unbounded metric distortion. 
\end{proposition}

In concurrent work by a subset of the authors~\cite{explainable}, it has been shown that Random Dictatorship has $\Theta(m^{1.5})$ utilitarian distortion, which is significantly worse than the optimal bound of $\Theta(\sqrt{m})$.

\begin{proposition}[\citealt{explainable}]
\label{prop:rd}
    Random Dictatorship has $\Theta(m^{1.5})$ utilitarian distortion. 
\end{proposition}

This leaves open the question of whether there exists a randomized rule that simultaneously achieves $\tilde{O}(\sqrt{m})$ utilitarian distortion and $O(1)$ metric distortion. 

Our proposed randomized voting rule, \RandMech, selects an agent uniformly at random and then selects an alternative with probability that is inversely proportional to its rank within the selected agent's preferences (i.e., drops harmonically), but with one caveat: it assigns higher probability to the alternative $\hat\X$ returned by the deterministic Plurality Veto voting rule and $0$ probability to all alternatives that the agent ranks below $\hat\X$.

\begin{definition}[\RandMech\ voting rule]
The \RandMech\ voting rule, $f^\text{TH}_\varepsilon$, is parameterized by a positive constant $\varepsilon$ and uses the following steps:
\begin{itemize}
    \item Let $\hat \X$ be the alternative that would be chosen by Plurality Veto 
    \item Select an agent $i\in\ags$ uniformly at random 
    \item  Return an alternative $\Y\in \alts$ with probability:
  $$\p(i, \Y) = \begin{cases}
\frac{\varepsilon}{6 H_mr_i(\Y)} & \text{if } \Y \succ_i \hat\X,\\
1-\sum_{\Y \succ_i \hat\X} \p(i, \Y) & \text{if } \Y = \hat\X, \\
0 & \text{otherwise.}
\end{cases}
$$

\end{itemize}
\label{def:randrule}
\end{definition}

Our main result for this section show that the \RandMech\ mechanism can essentially achieve the best of both worlds: 

\begin{theorem}
The \RandMech\ voting rule simultaneously guarantees a metric distortion of $3+\varepsilon$ and a utilitarian distortion of $O(\sqrt{m}H_m/\varepsilon)$.
\end{theorem}

We provide the proof of this theorem in the following two subsections, first for the metric world and then for the utilitarian one.

\subsection{Metric Distortion Guarantees}
We start off by defining the class of $\X$-truncated weight functions and proving a very useful lemma that extends Lemma 3 from \cite{AP17} to arbitrary $\X$-truncated weight functions (the previous result applied only to weight functions that assigned all their weight to the top alternative of each agent).

\begin{definition}
    Given an alternative $\X \in \alts$ and a preference profile $\prefp$, we call a function $\wf: \ags \times \alts \to [0, 1]$ an $\X$-truncated weight function ($\X$-TWF) if for each agent $i\in\ags$ it satisfies:
    \begin{itemize}
        \item For all $\Y\in\alts$ such that $\X \succ_i \Y$, the function assigns zero weight, i.e., $\wf(i, Y) = 0$,
        \item and for all other $\Y\in \alts$ the weight adds up to one: $\sum_{\substack{\Y \in \alts:\\ \Y \succeq_i \X}} \wf(i, Y) = 1$.
    \end{itemize}
    For such a function we let $\wf^+(\Y) = \sum_{i \in \ags} \wf(i, \Y)$ denote the total weight assigned to $\Y\in\alts$.
\end{definition}

\begin{lemma}
\label{lem:sc-up}

Given an alternative $\X \in \alts$, a preference profile $\prefp$, and any $\X$-truncated weight function $\wf_\X$, we have $\scf{\X}{d} \geq \frac{1}{2} \sum_{\Y \in \alts} \wf_\X^+(\Y)d(\X, \Y)$ for any metric $d \mcns \prefp$.
\end{lemma}

\begin{proof}
    \begin{align*}
        \scf{\X}{d} &= \sum_{i \in \ags} d(i, \X) \\ & = \sum_{i \in \ags} \sum_{\substack{\Y \in \alts:\\ \Y \succeq_i \X}} \wf_\X(i, \Y) d(i,\X) & (\because \forall i\in \ags \colon \sum_{\substack{\Y \in \alts:\\ \Y \succeq_i \X}}\wf_\X(i,\Y)=1 )\\
        & \geq \sum_{i \in \ags} \sum_{\substack{\Y \in \alts:\\ \Y \succeq_i \X}} \wf_\X(i, \Y) d(\X,\Y)/2 & (\text{by \Cref{lem:YW-half} and $\Y\succeq_i \X$})\\
        &= \sum_{i \in \ags} \sum_{\Y \in \alts} \wf_\X(i, \Y) d(\X,\Y)/2 & (\because \text{ $w(i,\Y)=0$ when $\X\succ_i\Y$})\\
        &= \sum_{Y \in \alts} \sum_{i \in \ags} \wf_\X(i, Y)d(X,Y)/2 \\ &= \sum_{Y \in \alts} \wf_\X^+(Y)d(X,Y)/2.\qedhere
    \end{align*}
\end{proof}

The next lemma generalizes Lemma 4 from \cite{AP17} which is restricted to the special case where the weight function assigns all its weight to each agent's top alternative.

\begin{lemma}
\label{thm:metric}
Consider any preference profile $\prefp$, any metric $d \mcns \prefp$, any alternative $\X\in\alts$, and any $\X$-truncated weight function $\wf_{\X}$. Then, if the metric distortion of $\X$ according to $d$ is $\md$, for any voting rule $f$  that assigns probability $p(\Y)$ to each alternative $\Y\in\alts$, we have

    $$\mdist(f(\prefp), d) \leq \md\left(1 + \frac{2n\sum_{Y \in \alts} p(\Y)  d(\X, \Y)}{\sum_{\Y \in \alts}  \wf^+_{\X}(\Y)d( \X, \Y)}\right).$$

\end{lemma}

\begin{proof}
If $\oc$ is the optimal alternative in the underlying metric space, we have:
\begin{align*}
    \mdist(f(\prefp), d) &= \frac{\E_{\Y \sim f(\prefp)}[\scf{\Y}{d}]}{\scf{\oc}{d}} \\
    & \leq \frac{\md\E_{\Y \sim f(\prefp)}[\scf{\Y}{d}]}{\scf{\X}{d}} & (\text{since $\scf{\X}{d}\leq \md\cdot \scf{\oc}{d}$})\\
    & = \frac{\md\sum_{\Y \in \alts} p(\Y) \scf{\Y}{d}}{\scf{\X}{d}} \\
    &\leq \frac{\md\sum_{\Y \in \alts} p(\Y) \left(\scf{\X}{d} + n \cdot d(\X, \Y)\right)}{\scf{\X}{d}} & (\text{by triangle inequality})\\
    &= \md\left(1 + \frac{n\sum_{\Y \in \alts} p(\Y)  d( \X, \Y)}{\scf{\X}{d}}\right) & (\text{since }\sum_{\Y \in \alts} p(\Y) = 1) \\
    &\leq \md\left(1 + \frac{n\sum_{Y \in \alts} p(\Y)  d(\X, \Y)}{\sum_{\Y \in \alts} \frac{1}{2} \wf^+_{\X}(\Y)d( \X, \Y)}\right) & (\text{by \Cref{lem:sc-up}}) \\
    &= \md\left(1 + \frac{2n\sum_{Y \in \alts} p(\Y)  d(\X, \Y)}{\sum_{\Y \in \alts} \wf^+_{\X}(\Y)d( \X, \Y)}\right).\qedhere
\end{align*}
\end{proof}

\begin{corollary}
The metric distortion of the \RandMech\ voting rule is $3+\varepsilon$.
\end{corollary}
\begin{proof}
First, let us define $\hat \X$-truncated weight function $\wfx$ as follows:

$$ \wfx(i, \Y) = \begin{cases}
\frac{1}{H_mr_i(\Y)} & \Y \succ_i \X,\\
1-\sum_{\Y \succ_i \X} \wfx(i, \Y) & \Y = \X, \\
0 & o.w.
\end{cases}.$$
Note that by the definition of $f^\text{TH}_\varepsilon$, we have $p(i, \Y) = \frac{\varepsilon}{6}\wfx(i, \Y)$ for $\Y \neq \hat \X$. In addition, for \emph{any} metric distance $d$ the alternative $\hat\X$ used by the \RandMech\ voting rule to define its truncated weight function has metric distortion at most $3$. Therefore, \Cref{thm:metric} implies that the distortion of the alternative chosen by this voting rule for any distance $d$, inducing preference profile $\prefp$, and using $\wfx$ as the truncated weight function, is:
        \begin{equation*}
    \mdist(f^\text{TH}_\varepsilon(\prefp),d)
     \leq 3 + \frac{6n\sum_{\Y \in \alts} p(\Y)  d(\hat \X, \Y)}{\sum_{\Y \in \alts}  \wfxp(\Y)d( \hat \X, \Y)},
     \end{equation*}
     where $p(\Y)$ is the probability that the \RandMech\ rule returns alternative $\Y$. Note that since we choose each agent $i$ uniformly at random, this probability is $p(\Y)=\frac{1}{n}\sum_i \p(i, \Y)$, so

    \begin{align*}
    \mdist(f^\text{TH}_\varepsilon(\prefp),d)
     & \leq 3 + \frac{6n\sum_{Y \in \alts} \frac{d(\hat \X, \Y) \sum_{i \in \ags} p(i, \Y) }{n}}{\sum_{\Y \in \alts}  \wfxp(\Y)d( \hat \X, \Y)} \\
     & = 3 + \frac{6\sum_{Y \neq \hat \X}d(\hat \X, \Y) \sum_{i \in \ags} p(i, \Y)}{\sum_{\Y \in \alts}  \wfxp(\Y)d( \hat \X, \Y)}  & (d(\X, \X) = 0)\\
     & = 3 + \frac{6\sum_{Y \neq \hat \X} \frac{\varepsilon}{6} \wfxp(\Y) d(\hat \X, \Y)}{\sum_{\Y \in \alts}  \wfxp(\Y)d( \hat \X, \Y)} & (p(i, \Y) = \frac{\varepsilon}{6}\wfx(i, \Y))  \\
    & = 3 + \varepsilon.\qedhere
\end{align*}

\end{proof}

\subsection{Utilitarian Distortion Guarantees}

\begin{theorem}
    The utilitarian distortion of the \RandMech\ voting rule is $O(\sqrt m H_m)$.
\end{theorem}

\begin{proof}
Consider any preference profile $\prefp$ and any utility profile $\utp \ucns \prefp$. Let $\oc$ be the optimal alternative according to $\utp$. We partition the agents into two sets $\ags_1 = \set{i \in \ags: \hat \X \succ_i X^*}$,
and $\ags_2 = \ags \setminus \ags_1$. For each $k\in\{0,1\}$, we use $\swfs{k}{\Y}{\utp} = \sum_{i \in \ags_k} u_i(\Y)$ to denote the total value of alternative $\Y$ for the agents in $\ags_k$, and let $\Su = \sum_{i \in \ags} \sum_{\Y \succ_i \hat \X } u_i(\Y)$ be the sum of the utilities of the agents for the alternatives they prefer to $\hat \X$.
    In addition for alternative $\Y$ let $\swfs{\su}{\Y}{\utp}$ be the utility that she gains when she appears above $\hat \X$, i.e., $  \swfs{\su}{\Y}{\utp} = \sum_{\substack{i \in \ags, \\ \Y \succ_i \hat \X}} u_i(\Y).$
    Note that $p(\Y) \geq \nicefrac{\varepsilon \swfs{\su}{\Y}{\utp}}{(6n H_m)}$, and the probability of $\hat \X$ being selected by this rule is at least $1-\ex$. We have:

\begin{align*}
    \E_{\Y \sim f^\text{TH}_\varepsilon(\prefp)}[\swf{\Y}{\utp}] & \geq \left(1-\ex\right)\swf{\hat \X}{\utp} + \sum_{\Y \neq \hat \X} p(\Y)\swf{\Y}{\utp} \\
    &\geq \left(1-\ex\right)\swf{\hat \X}{\utp} + \sum_{\Y \neq \hat \X} \frac{\varepsilon \swfs{\su}{\Y}{\utp}}{6 nH_m}\swf{\Y}{\utp} \\
    &\geq \left(1-\ex\right)\swf{\hat \X}{\utp} + \sum_{\Y \neq \hat \X} \frac{\varepsilon \swfs{\su}{\Y}{\utp}^2}{6nH_m} \\
    &\geq \left(1-\ex\right)\swf{\hat \X}{\utp} + \frac{\varepsilon}{6nmH_m} \left(\sum_{\Y \neq \hat \X} \swfs{\su}{\Y}{\utp} \right)^2\\
    & = \left(1-\ex\right)\frac{n-\Su}{m} + \frac{\varepsilon \Su^2}{6nmH_m} & (\because \text{Decreasing in terms of } S)\\
     &\geq \frac{\varepsilon n}{6mH_m}.
\end{align*}
That means 
\begin{equation}
    \udist(f^\text{TH}_\varepsilon(\prefp), \utp) \leq \frac{\swf{\oc}{\utp}}{\nicefrac{\varepsilon n}{(6mH_m)}} = \frac{6mH_m\swf{\oc}{\utp}}{\varepsilon n}.
    \label{eq:udist-sw}
\end{equation}
On the other hand, we have 
    \begin{align*}
        \udist(f^\text{TH}_\varepsilon(\prefp), \utp) &=\frac{\swf{\oc}{\utp}}{\E_{\Y \sim f^\text{TH}_\varepsilon(\prefp)}[\swf{\Y}{\utp}]} \\
        &=\frac{\swfs{1}{\oc}{\utp}}{\E_{\Y \sim f^\text{TH}_\varepsilon(\prefp)}[\swf{\Y}{\utp}]} + \frac{\swfs{2}{\oc}{\utp}}{\E_{\Y \sim f^\text{TH}_\varepsilon(\prefp)}[\swf{\Y}{\utp}]} \\
        &\leq \frac{\swfs{1}{\X}{\utp}}{\left(1-\ex\right) \swf{\X}{\utp}} + \frac{\swfs{2}{\oc}{\utp}}{\E_{\Y \sim f^\text{TH}_\varepsilon(\prefp)}[\swf{\Y}{\utp}]} \\
        &\leq \frac{6}{6-\varepsilon} + \frac{\swfs{2}{\oc}{\utp}}{\frac{\varepsilon \swfs{\su}{\Y}{\utp}}{6n H_m}\swf{\oc}{\utp}}\\
    & = \frac{6}{6-\varepsilon} + \frac{6n H_m}{\varepsilon \swf{\oc}{\utp}} &\left(\because \swfs{2}{\oc}{\utp} = \swfs{s}{\oc}{\utp}\right) \\
    & \leq \frac{12 n H_m}{\varepsilon \swf{\oc}{\utp}}.
\end{align*}

Finally, if we put it together with \Cref{eq:udist-sw}, and use the fact that $\min(a,b) \le \sqrt{a \cdot b}$, we have $$\udist(f^\text{TH}_\varepsilon(\prefp), \utp) \leq \sqrt{\frac{12 n H_m}{\varepsilon \swf{\oc}{\utp}} \times \frac{6mH_m\swf{\oc}{\utp}}{\varepsilon n}}.$$
Since this holds for any $\prefp$ and $\utp \ucns \prefp$, and $\varepsilon$ is a positive constant, we get that $\udist(f^\text{TH}_\varepsilon)  \in O(\sqrt {m} H_m).$ 
\end{proof}

\section{Partial Ordinal Preferences (Top-$\mathbf{t}$)}

In this section, we consider the problem of designing a voting rule with a good distortion in both worlds while we only have each agent's ranking for their top-$t$ preferences, where $t<m$.

\paragraph{Top-$t$ distortion.} Similar to the full ranking case, we can define the distortion of a voting rule $f$ on top-$t$ preference profile $\prefp_t$ as the worst-case distortion of $f(\prefp_t)$ on any metric (utility) profile that is consistent with $\prefp_t$. Note that since $\prefp_t$ contains less information than a full-ranking preference profile, a wider class of metrics (utility profiles) are considered in the definition of distortion.

\subsection{Deterministic Rules with Top-$t$ Preferences}
When it comes to deterministic voting rules with top-$t$ preferences, we can achieve $\Theta(m^2)$ utilitarian distortion by choosing the plurality winner, which is optimal. In the metric framework, the distortion of any deterministic voting rule is at least $2m/t$. \cite{kempe2020communication} and \citet{AFP22} propose a deterministic voting rule with $6m/t+1$ metric distortion. We first show that achieving the best of both worlds in this case is infeasible.

\begin{theorem}
    Any deterministic voting rule on top-$t$ preferences with bounded utilitarian distortion has metric distortion of $\Omega(m-t+1)$.\end{theorem}

\begin{proof}
    Consider the preference profile $\prefp$ where $m-t+1$ of the alternatives each appear as the top choice of $\nicefrac{n}{(m-t+1)}$ agents, and the remaining $t-1$ alternatives appear with the same order in the second to the $t$\textsuperscript{th} position of all the agents. Let $\X$ be the selected candidates. If $X$ does not appear as the top choice of any agent, then in the utility profile where each agent has utility 1 for his top choice and zero for the rest we have unbounded distortion. Now let us consider the case where $X$ has plurality score greater than zero. In that case consider the 1 dimensional metric space $d$ where $X$ is located at point 1, all the agents who have her as their top choice are located at point $0.5$, and all remaining agents and alternatives are located at point zero. 
    You can see that $X$ has distance of at lest $0.5$ to each agent and hence $\scf{X}{d} > \nicefrac{n}{2}$ and for any $\Y \neq \X$, $\scf{Y}{d} = \nicefrac{n}{2(m-t+1)}$
    which means $X$ has $\Omega(m-t+1)$ metric distortion.
\end{proof}

We remark the contrast between this metric distortion lower bound of $\Omega(m-t+1)$ when bounded utilitarian distortion is imposed, and the optimal metric distortion of $\Theta(m/t)$ when we ignore utilitarian distortion~\cite{kempe2020communication}. This shows that requiring utilitarian distortion to be bounded makes metric distortion strictly worse whenever $t$ is neither $\Theta(1)$ nor $m-\Theta(1)$.

Next, we show that this lower bound is tight by proving a matching upper bound. 
\begin{theorem}
    A deterministic voting rule exists with metric distortion $O(m-t+1)$ and utilitarian distortion $O(m^2)$.
    \label{thm:top-t-det}
\end{theorem}

\begin{proof}
Consider any preference profile $\prefp$, utility profile $\utp \ucns \prefp$, and metric $d\mcns \prefp$. If $t \leq \nicefrac{m}{2}$, we can use the Plurality rule which has $O(m^2)$ utilitarian and $O(m) = O(m-t+1)$ metric distortion, giving us the desired bounds.

Suppose that $t > m/2$, and let $\alts^+$ be the set of alternatives that appear as the top choice of at least $n/2m$ agents. If $|\alts^+| < 2(m-t+1)$, there exists an alternative $\X$ that appears as the top choice of at least $\nicefrac{n}{4(m-t+1)}$ agents. This alternative at least gets a $1/m$ utility for each time she appears as the top choice of an agent, so her social welfare would at least be $\nicefrac{n}{(4m(m-t+1))}$, and since the maximum social welfare of any alternative is $n$ this gives us $O(m(m-t+1))$ utilitarian distortion. In addition, we know that an alternative that appears as the top choice of $n'$ agents at most has a metric distortion of $\frac{n-n'}{n}$ that means $\X$ at most has a metric distortion of $ \frac{n-\nicefrac{n}{(4m(m-t+1))}}{\nicefrac{n}{(4m(m-t+1))}} \in O(m-t+1) $.

In the other case, where $|\alts^+| \geq 2(m-t+1)$, $\alts^+$ includes at least $|\alts^+|-m+t$ of the top-$t$ alternatives of each agent's list. That means we can run a top-$|\alts^+|-m+t$ deterministic voting rule to select alternative $\X$ with $O(\frac{|\alts^+|}{|\alts^+|-m+t})$ metric distortion (by Theorem 4.5 of \cite{AFP22}). Note that this distortion is only compared to the best alternative in $\alts^+$ and not compared to the best alternative overall. Since each member of $\alts^+$ has a plurality score of at least $n/2m$, we can apply \Cref{lem:subset-mdist} with $\tau = 1/2m$ and show that the $\mdist(\X, d)\in O(\frac{|\alts^+|}{|\alts^+|-m+t}) = O(1)$. Finally, since $\X \in |\alts^+|$, she has at least $n/2m$ top votes and at least gets $1/m$ utility for each of them. That means her social welfare is $O(n/m^2)$, and since the maximum possible social welfare is $n$, this yields $O(m^2)$ utilitarian distortion.
\end{proof}

\begin{remark}
For each instance, if the mechanism described in the proof of \Cref{thm:top-t-det} outputs an alternative with $\md$ metric and $\ud$ utilitarian distortion, we have $\md \times \ud \in O\left(\max \left(m (m-t+1)^2, m^2\right)\right)$.
\end{remark}

\subsection{Randomized Rules with Top-$t$ Preferences}
In the randomized case, we know that we can achieve $O(\max(\sqrt m, m/t))$ utilitarian and $3-2/n$ metric distortion. The following theorem shows that for $t < \sqrt m$, we cannot simultaneously achieve the (asymptotically) best possible distortion in both worlds.

\begin{theorem}
Any (randomized) voting rule with access only to the top-$t$ prefix of each agent's ranking that achieves a metric distortion of $\md$ will have a utilitarian distortion of $\Omega(\frac{m\sqrt m}{\md t \sqrt t})$.
\end{theorem}

\begin{proof}
    We will construct a top-$t$ preference profile $\prefp$ with $n$ agents and $m$ alternatives. Assume that $n$ is divisible by $\frac{\md m \sqrt m}{t \sqrt t}$ and $m$ is divisible by $3t$. We partition the alternatives into $\frac{m}{3t} + 2$ sets: let $\alts^+$ and $\alts^-$ each include $m/3$ alternatives and, for $k \in [\frac{m}{3t}]$, let $\alts_k$ include $t$ alternatives. Choose any $X \in \alts^-$. 
    
    In $\prefp$, for $k \in [\frac{m}{3t}]$, let $\ags_k$ be a set of $\frac{n t \sqrt t}{\md m \sqrt m}$ agents who rank alternatives in $\alts_k$ as their top-$t$ choices (in any order). The remaining agents rank alternatives in $\alts^+$ as their top-$t$ choices in such a way that each of $m/3$ alternatives in $\alts^+$ appears in the top-$t$ choices of a $\nicefrac{t}{(m/3)} = 3t/m$ fraction of these agents.    Consider any voting rule $f$ with a metric distortion at most $\md$ on this preference profile, where $\md$ is finite. We want to prove that its utilitarian distortion must be $\Omega(\frac{m\sqrt m}{\md t \sqrt t})$. 
    
    First, it is easy to check that $f(\prefp)$ cannot assign any positive probability to any alternative in $\alts^-$. This is because it is possible that all agents rank the alternatives in $\alts^-$ at the bottom of their preference rankings, so $\prefp$ is consistent with distance metrics in which every alternative in $\alts^-$ has an arbitrarily large social cost. Thus, assigning any positive probability to any alternative in $\alts^-$ would result in unbounded metric distortion. 
    
    Next, define $p_k$ to be the sum of the probabilities assigned to the alternatives in $\alts_k$ under $f(\prefp)$, and let $k^* \in \argmax_{k \in [\frac{m}{3t}]} p_k$. 
        Now, consider the simple one-dimensional metric $d$ over $\bbR$ in which alternatives in $\alts_{k^*}$ are at $0$, agents in $\ags_{k^*}$ are at $1$, and all other agents and alternatives are at $2$. Hence, the distances are as follows for all agents $i$ and alternatives $\Y$:
    $$
    d(i, \Y) = \begin{cases}
        1 & i \in \ags_{k^*} \\
        2 & i \notin \ags_{k^*} \text{ and } \Y \in \alts_{k^*} \\
        0 & i \notin \ags_{k^*} \text{ and } \Y \notin \alts_{k^*}.
    \end{cases}.
    $$

    For the metric distortion of $f$, we have:
    \begin{align*}
        \mdist(f(\prefp), d) &= \frac{\E_{\Y \sim f(\prefp)}[\scf{\Y}{d}]}{\min_{\Y \in \alts} \scf{\Y}{d}} \\
        & \geq \frac{p_{k^*} (2n-\frac{n t \sqrt t}{\md m \sqrt m}) + (1-p_{k^*}) \min_{\Y \in \alts} \scf{\Y}{d}}{\min_{\Y \in \alts} \scf{\Y}{d}} \\
        &= 1-p_{k^*} + \frac{p_{k^*} (2n-\frac{n t \sqrt t}{\md m \sqrt m})}{\frac{n t \sqrt t}{\md m \sqrt m}} \\
        &\geq 1-p_{k^*} + \frac{p_{k^*} (2-\frac{t \sqrt t}{\md m \sqrt m})}{\frac{t \sqrt t}{\md m \sqrt m}} & (\varepsilon \to 0) \\
        = 1-p_{k^*} + p_{k^*} \cdot \left(\frac{2\md m \sqrt m}{t \sqrt t}-1\right)\\
        &= 1+2p_{k^*} \cdot \left(\frac{\md m \sqrt m}{t\sqrt t}-1 \right) \geq 1+2p_{k^*} \cdot \frac{(\md -1 ) m \sqrt m}{t\sqrt t}.
    \end{align*}
   Since $f$ has metric distortion at most $\md$, we know that  $\mdist(f(\prefp), d) \leq \md$, which means
    \begin{equation}
        2p_{k^*} \cdot \frac{(\md-1) m \sqrt m}{t\sqrt t} \leq \md - 1 \implies p_{k^*} \leq  \frac{t \sqrt t}{2 m \sqrt m}.
        \label{eq:pk}
    \end{equation}
    
    We use this to derive a lower bound on the utilitarian distortion of $f$. Suppose that in the underlying full preference profile, every agent has alternatives in $\alts^-$ at ranks $t+1$ through $t+m/3$, with $\X$ appearing at rank $t+1$. Consider the consistent utility profile $\utp$ in which agents who have an alternative in $\alts^+$ as their top choice have a utility of $\frac{1}{t+m/3}$ for their top $t+m/3$ alternatives and the remaining agents have a utility of $\frac{1}{t+1}$ for their top $t+1$ alternatives.

    Under this utility profile, we have the following:
    \begin{align*}
    \swf{\X}{\utp} &\ge \frac{1}{t+1} \cdot \frac{n\sqrt t}{3\md \sqrt m} \ge \frac{n}{6 \md \sqrt{mt}},&&\\
    \swf{\Y}{\utp} &\le \frac{n \cdot 3t/m}{t+m/3} \leq \frac{9nt}{m^2}, &&\forall \Y \in \alts^+,\\
    \swf{\Y}{\utp} &= \frac{1}{t+1} \cdot \frac{nt\sqrt t}{\md m \sqrt m} \le \frac{n\sqrt{t}}{\md m\sqrt{m}}, &&\forall \Y \in \alts_k, k \in [m/(3t)].
    \end{align*}
    Additionally, since $p_k \leq p_{k^*}$ for $k \in [m/(3t)]$, we have 

    \begin{equation*}
      \mdist(f(\prefp), \utp) \ge \frac{\swf{\X}{\utp}}{\E_{\Y \sim f(\prefp)}[\swf{\Y}{\utp}]} \geq \frac{\frac{n}{6 \md \sqrt{mt}}}{p_{k^*} \cdot \frac{m}{3t} \cdot \frac{n\sqrt{t}}{\md m\sqrt{m}} + \frac{9nt}{m^2} } 
       \geq \frac{\frac{n}{6 \md \sqrt{mt}}}{\frac{nt}{6 \md m^2} + \frac{9nt}{m^2} }  \in \Omega\left(\frac{m \sqrt m}{\md t \sqrt t} \right),
    \end{equation*}
where the last inequality is implied by  \Cref{eq:pk}.
\end{proof}

Recall that Random Dictatorship (which only requires access to plurality votes) achieves metric distortion less than $3$~\cite{AP17}, so the optimal metric distortion of randomized rules given top-$t$ preferences is $O(1)$. The following corollary shows that any rule achieving this optimal metric distortion bound must have utilitarian distortion $\Omega(\max((m/t)^{1.5},\sqrt{m}))$, which is strictly worse than the optimal bound of $O(\max(m/t,\sqrt{m}))$~\cite{BHLS22} when $t = o(m^{2/3})$. 

\begin{corollary}
    Any voting rule with constant metric distortion on top-$t$ preference profiles has utilitarian distortion of $\Omega\left(\max\left(\frac{m\sqrt m}{t\sqrt t},\sqrt{m}\right) \right)$.
\end{corollary}

\begin{definition}[\TRandMech\ voting rule]
The \TRandMech\ voting rule, $f^\text{TTH}$ uses the following steps:
\begin{itemize}
    \item Let $\hat \X$ be the alternative that would be chosen by the mechanism proposed by \citet{AFP22}, which has $\nicefrac{6m}{t}+1$ metric distortion.
    \item Select an agent $i\in\ags$ uniformly at random. 
    \item  Return an alternative $\Y \in \alts$ with probability:
\[\p(i, \Y) = \begin{cases}
\frac{1}{2H_tr_i(\Y)} & \text{if } \Y \succ_i \hat\X,\\
1-\sum_{\Y \succ_i \hat\X} \p(i, \Y) & \text{if } \Y = \hat\X, \\
0 & \text{otherwise.}
\end{cases}
\] 
\end{itemize}
\label{def:toptrandrule}
\end{definition}

\begin{theorem}
    The \TRandMech\ rule has metric distortion of $O\left(\frac{m}{t}\right)$.
\end{theorem}

\begin{proof}
First, let us define $\hat \X$-truncated weight function $\wfx$ as $\wfx(i, \Y) = p(i, \Y)$.

Note that for \emph{any} metric distance $d$ the alternative $\hat\X$ used by the \TRandMech\ voting rule to define its truncated weight function has metric distortion at most $7m/t$. Therefore, \Cref{thm:metric} implies that the distortion of the alternative chosen by this voting rule for any distance $d$, inducing preference profile $\prefp$, and using $\wfx$ as the truncated weight function, is:
        \begin{equation*}
    \mdist(f^\text{TTH}_\varepsilon(\prefp),d)
     \leq \frac{7m}{t} + \frac{14mn\sum_{\Y \in \alts} p(\Y)  d(\hat \X, \Y)}{t \sum_{\Y \in \alts}  \wfxp(\Y)d( \hat \X, \Y)},
     \end{equation*}
     where $p(\Y)$ is the probability that the \TRandMech\ rule returns alternative $\Y$. Note that since we choose each agent $i$ uniformly at random, this probability is $p(\Y)=\frac{1}{n}\sum_i \p(i, \Y) = \frac{1}{n}\sum_i \wfx(i, \Y)$, which yields:
    \begin{align*}
    \mdist(f^\text{TTH}_\varepsilon(\prefp),d) & \leq \frac{7m}{t} + \frac{14mn\sum_{\Y \in \alts} \frac{\wfxp(\Y)}{n}  d(\hat \X, \Y)}{t \sum_{\Y \in \alts}  \wfxp(\Y)d( \hat \X, \Y)} \leq \frac{21m}{t} \in O\left(\frac{m}{t}\right).\qedhere
\end{align*}
\end{proof}

\begin{theorem}
    The \TRandMech\ rule has utilitarian distortion of $O\left(\frac{m \sqrt m H_t}{t} \right)$.
\end{theorem}

\begin{proof}
Consider any preference profile $\prefp$ and utility profile $\utp \ucns \prefp$. Let $p$ be the distribution returned by the Top-$t$ Truncated Harmonic rule on $\prefp$. Let $\oc \in \max_{\X \in \alts} \swf{\X}{\utp}$ be an optimal alternative under $\utp$. Define $T_i$ to be the set of top-$t$ alternatives of agent $i$. Let $\Su_1 = \sum_{i \in \ags} \sum_{\Y \in T_i : \Y \succ_i \hat \X } u_i(\Y)$
be the sum of the utilities of the agents for the alternatives in their top-$t$ choices that they prefer to $\hat \X$, and $\Su_2$ be the sum of the utilities of the agents for all of their top-$t$ choices.

In addition, for alternative $\Y$, let $\swfs{\su}{\Y}{\utp}$ be her total utility from the agents who rank her among their top-$t$ choices and above $\hat \X$, i.e., $  \swfs{\su}{\Y}{\utp} = \sum_{i \in \ags : \Y \succ_i \hat \X} u_i(\Y)$.
    
    Note that $p(\Y) \geq \frac{\swfs{\su}{\Y}{\utp}}{2n H_t}$ for each $\Y \in \alts \setminus \set{\hat{\X}}$ and $p(\hat \X) \ge 1/2$, so
\begin{align*}
    \E_{\Y \sim f^\text{TTH}(\prefp)}[\swf{\Y}{\utp}] &= \frac{1}{2} \swf{\hat \X}{\utp} + \sum_{\Y \neq \hat \X} p(\Y)\swf{\Y}{\utp} \\
    &\geq \frac{1}{2} \swf{\hat \X}{\utp} + \sum_{\Y \neq \hat \X} \frac{\swfs{\su}{\Y}{\utp}}{2n H_t} \cdot \swfs{\su}{\Y}{\utp} \\
    &\geq \frac{1}{2} \swf{\hat \X}{\utp} + \frac{\left(\sum_{\Y \neq \hat \X} \swfs{\su}{\Y}{\utp}\right)^2}{2n m H_t} & (\because \text{AM-QM inequality})\\
     &\geq \frac{\Su_2-\Su_1}{2m}  + \frac{\Su_1^2}{2n m H_t} \geq \frac{\Su_2^2}{4n m H_t} \geq \frac{nt^2}{4m^3 H_t}, & (\because  \Su_2 \geq nt/m)
\end{align*}
where the fifth transition holds because $\frac{\Su_2-\Su_1}{2m}+\frac{\Su_1^2}{2nmH_t} \ge \frac{\Su_2-\Su_1}{2m}+\frac{\Su_1^2}{4nmH_t} \ge \frac{\Su_2^2}{4nmH_t}$ (the last transition is equivalent to $\Su_1+\Su_2 \le 2nH_t$, which is true because $\Su_1+\Su_2 \le 2n$ trivially). Next, let $\plu(\Y)$ be the number of agents whose top choice is $\Y$. We have
\begin{align*}
    \E_{\Y \sim f^\text{TTH}(\prefp)}[\swf{\Y}{\utp}] &\geq \sum_{\Y \in \alts} \frac{\plu(\Y)}{2nH_t} \cdot \swf{\Y}{\utp} \geq \frac{1}{2nH_t} \sum_{\Y \in \alts} \frac{\plu(\Y)^2}{m} \\
    &\geq \frac{1}{2nH_t} \left(\sum_{\Y \in \alts}\frac{\plu(\Y)}{m}\right)^2
    \ge \frac{n}{2m^2H_t}.
\end{align*}
That means
\begin{equation}
    \udist(f^\text{TTH}(\prefp), \utp) \leq \frac{\swf{\oc}{\utp}}{\max\left(\frac{nt^2}{4 m^3 H_t}, \frac{n}{2m^2H_t}\right)} = \min\left(\frac{4m^3 H_t \swf{\oc}{\utp}}{nt^2}, \frac{2m^2 H_t \swf{\oc}{\utp}}{n}\right).
    \label{eq:udist-sw-t}
\end{equation}

Now, let $\ags_1$ be the set of agents that have $\oc$ in their top-$t$ choices and prefer $\oc$ to $\hat \X$, $\ags_2$ be the set of agents that have $\hat \X$ in their top-$t$ choices and prefer $\hat \X$ to $\oc$, and $\ags_3$ be the remaining agents. Also, for $k \in \set{1,2,3}$, let $P_k$ be the sum of utilities of agents in $\ags_k$ for $\oc$. Note that $\swf{\oc}{\utp} = P_1+P_2+P_3$. 

If $P_3 \geq P_1 + P_2$ we have $\swf{\oc}{\utp} \le 2P_3$. On the other hand,
\begin{align*}
\E_{\Y \sim f^\text{TTH}(\prefp)}[\swf{\Y}{\utp}] &\geq \sum_{\Y \in \alts} p(\Y) \cdot \sum_{i \in \ags_3 : \Y \in T_i} u_i(\Y) \geq  \sum_{\Y \in \alts}\frac{ \left( \sum_{i \in \ags_3, \Y \in T_i} u_i(\Y)\right)^2}{2nH_t}\\
&\geq \frac{ \left( \sum_{\Y \in \alts}\sum_{i \in \ags_3, \Y \in T_i} u_i(\Y)\right)^2}{2nmH_t} \geq \frac{t^2 P_3^2}{2n m H_t} \ge \frac{t^2 \swf{\oc}{\utp}^2}{8n m H_t},
\end{align*}
where the third transition is due to the AM-QM inequality. Thus, we get that $\udist(f^\text{TTH}(\prefp), \utp) \le \frac{8nm H_t}{t^2\swf{\oc}{\utp}}$. On the other hand, if $P_3 \le P_1+P_2$, we have:
\begin{align*}
    \udist(f^\text{TTH}(\prefp), \utp) &= \frac{\swf{\oc}{\utp}}{\E_{\Y \sim f^\text{TTH}(\prefp)}[\swf{\Y}{\utp}]} \leq \frac{2P_2}{\frac{1}{2}\swf{\hat \X}{\utp}} + \frac{2P_1}{p(\oc)\swf{\oc}{\utp}}\\
    &\leq  \frac{4P_2}{P_2} + \frac{2P_1}{\frac{P_1}{2n H_t}\swf{\oc}{\utp}} \leq  4 + \frac{4n H_t}{\swf{\oc}{\utp}} \leq \frac{8n H_t}{\swf{\oc}{\utp}}.
\end{align*}
That means $$
    \udist(f^\text{TTH}(\prefp), \utp)  \leq \max\left(\frac{8nm H_t}{t^2\swf{\oc}{\utp}}, \frac{8n H_t}{\swf{\oc}{\utp}}\right).$$
Finally, combined with \Cref{eq:udist-sw-t}, and using the fact that $\min(a,b) \le \sqrt{a \cdot b}$, for $t \leq \sqrt m$ we have $$\udist(f^\text{TTH}(\prefp), \utp) \leq \sqrt{\frac{4nm H_t}{t^2\swf{\oc}{\utp}}\times \frac{2m^2 H_t \swf{\oc}{\utp}}{n}} \in O\left(\frac{m \sqrt m H_t}{t}\right),$$
and for $t \geq \sqrt m$ we have
$$\udist(f^\text{TTH}(\prefp), \utp) \leq \sqrt{\frac{8n H_t}{\swf{\oc}{\utp}}\times \frac{4m^3 H_t \swf{\oc}{\utp}}{nt^2}} \in O\left(\frac{m \sqrt m H_t}{t}\right).\qedhere$$
\end{proof}

The Top-$t$ Truncated Harmonic rule provides one possible tradeoff by achieving $O(m/t)$ metric distortion and $O(m\sqrt{m} H_t / t)$ utilitarian distortion. We already know that Random Dictatorship, which can be executed given only top-$1$ preferences (and thus, also given top-$t$ preferences for any $t \ge 1$) provides a different tradeoff by achieving $O(1)$ metric distortion~\cite{AP17} and $O(m^{1.5})$ utilitarian distortion (\Cref{prop:rd}). Using the next result, one can combine these rules (or any two rules with different tradeoffs) to find a spectrum of possible tradeoffs, which may help in eventually charting out the Pareto frontier.

\begin{theorem}
\label{thm:rule_combination}
Fix any $\beta \in [0,1]$. If we have voting rule $f^1$ with metric distortion $\md_1$ and utilitarian distortion $\ud_1$, and voting rule $f^2$ with metric distortion $\md_2$ and utilitarian distortion $\ud_2$, we can combine these two rules to achieve $\beta \md_1 + (1-\beta) \md_2 $ metric distortion and $\frac{\ud_1 \ud_2}{\beta \ud_2 + (1- \beta) \ud_1} \leq \max\{\nicefrac{\ud_1}{\beta} , \nicefrac{\ud_2}{(1-\beta)}\}$ utilitarian distortion.
\end{theorem}
\begin{proof}
    Consider preference profile $\prefp$, utility profile $\utp \ucns \prefp$, and metric $d \mcns \prefp$. Define rule $f$ to run rule $f^1$ with probability $\beta$ and rule $f^2$ with probability $1-\beta$. For the metric distortion we have:
    \begin{align*}
        \mdist(f(\prefp), d) &= \frac{\E_{\X \sim f(\prefp)}[\scf{\X}{d}]}{\min_{\X \in \alts}\scf{\X}{d} } \\
        &=\frac{\beta \E_{\X \sim f^1(\prefp)}[\scf{\X}{d}] + (1-\beta) \E_{\X \sim f^2(\prefp)}[\scf{\X}{d}]}{\min_{\X \in \alts}\scf{\X}{d} } \\
        &= \frac{\beta \E_{\X \sim f^1(\prefp)}[\scf{\X}{d}]}{\min_{\X \in \alts}\scf{\X}{d} } +\frac{(1-\beta) \E_{\X \sim f^2(\prefp)}[\scf{\X}{d}]}{\min_{\X \in \alts}\scf{\X}{d} } \\
        &= \beta \mdist(f^1(\prefp), d) + (1-\beta)\mdist(f^2(\prefp), d) \le \beta \md_1 + (1-\beta) \md_2.
    \end{align*}

On the other hand, for the utilitarian case we have:
\begin{align*}
        \udist(f(\prefp), \utp) &= \frac{\max_{\X \in \alts}\swf{\X}{\utp} }{\E_{\X \sim f(\prefp)}[\swf{\X}{d}]}\\
        &= \frac{\max_{\X \in \alts}\swf{\X}{\utp} }{\beta \E_{\X \sim f^1(\prefp)}[\swf{\X}{d}] + (1-\beta) \E_{\X \sim f^2(\prefp)}[\swf{\X}{d}]}\\
        &=\frac{1}{\beta \frac{\E_{\X \sim f^1(\prefp)}[\swf{\X}{d}]}{\max_{\X \in \alts}\swf{\X}{\utp} } + (1-\beta) \frac{\E_{\X \sim f^2(\prefp)}[\swf{\X}{d}]}{\max_{\X \in \alts}\swf{\X}{\utp} }} \\
        &\le \frac{1}{\frac{\beta}{\ud_1} + \frac{1-\beta}{\ud_2}} = \frac{\ud_1 \ud_2}{\beta \ud_2 + (1- \beta) \ud_1}.
    \end{align*}
\end{proof}

\section{Discussion}
In this work, we investigate the feasibility of achieving asymptotically (near-)optimal distortion under both metric and utilitarian worlds simultaneously. We prove that while this is possible given full rankings, this is not the case given partial rankings in the form of top-$t$ preferences.

Our work leaves open a number of exciting open questions. For example, our deterministic rule, Pruned Plurality Veto, for achieving the best of both worlds given full rankings achieves a metric distortion of $9+\epsilon$. Even though this is a constant, it is undesirably high, and improving it is an important direction. We prove that it cannot be improved all the way to $3$, but our lower bound is only slightly greater than $3$. Achieving tight distortion bounds for randomized rules with top-$t$ preferences is another obvious open direction. For example, does there exist a randomized voting rule with top-$t$ preferences achieving constant metric distortion and $O(\frac{m\sqrt m}{t \sqrt m})$ utilitarian distortion?

More broadly, it would be interesting to derive ``best of both worlds'' style results for other settings studied in the distortion literature, such as the utilitarian distortion with unit-range instead of unit-sum assumption (where the minimum and maximum utility of each agent is set to $0$ and $1$, respectively, instead of normalizing the sum to $1$), alternative ballot designs such as approval ballots, and returning a committee or a ranking of alternatives instead of a single alternative.

Another compelling direction concerns the communication complexity of these rules. Our results indicate that it is infeasible to simultaneously achieve asymptotically optimal distortion in both worlds given top-$t$ preferences, for some values of $t$.  However, it is known that one can achieve lower utilitarian distortion while eliciting fewer bits of information than under top-$t$ preferences by using more efficient elicitation methods~\cite{mandal2019thrifty,mandal2020optimal}. This line of work studies the communication-distortion tradeoff, i.e., the optimal distortion achievable while eliciting a given number of bits of information regarding agent preferences. Combining it with our best of both worlds guarantee adds another axis of metric versus utilitarian to that tradeoff.  

\section*{Acknowledgments}
Gkatzelis was partially supported by NSF CAREER award CCF 2047907. Latifian and Shah were partially supported by an NSERC Discovery Grant.

\newpage
\bibliographystyle{ACM-Reference-Format}
\bibliography{abb,bestofboth}

\end{document}